\documentclass[sigconf]{aamas}

\usepackage{balance} 

\doi{CZFV1739}

\makeatletter
\gdef\@copyrightpermission{
  \begin{minipage}{0.2\columnwidth}
   \href{https://creativecommons.org/licenses/by/4.0/}{\includegraphics[width=0.90\textwidth]{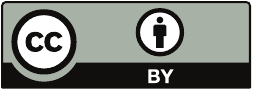}}
  \end{minipage}\hfill
  \begin{minipage}{0.8\columnwidth}
   \href{https://creativecommons.org/licenses/by/4.0/}{This work is licensed under a Creative Commons Attribution International 4.0 License.}
  \end{minipage}
  \vspace{5pt}
}
\makeatother

\setcopyright{ifaamas}
\acmConference[AAMAS '26]{Proc.\@ of the 25th International Conference
on Autonomous Agents and Multiagent Systems (AAMAS 2026)}{May 25 -- 29, 2026}
{Paphos, Cyprus}{C.~Amato, L.~Dennis, V.~Mascardi, J.~Thangarajah (eds.)}
\copyrightyear{2026}
\acmYear{2026}
\acmDOI{}
\acmPrice{}
\acmISBN{}

\title[Truthful Reporting of Competence]{Truthful Reporting of Competence with Minimal Verification}

\author{Reshef Meir}
\affiliation{
  \institution{Technion---Israel Institute of Technology}
  \city{Haifa}
  \country{Israel}}
    \orcid{0000-0003-0961-3965}
\email{reshefm@technion.ac.il}

\author{Jonathan Wagner}
\affiliation{
   \institution{Technion---Israel Institute of Technology}
  \city{Haifa}
  \country{Israel}}
    \orcid{0009-0001-9022-7967}
\email{swagner@campus.technion.ac.il}

\author{Omer Ben{-}Porat}
\affiliation{
   \institution{Technion---Israel Institute of Technology}
  \city{Haifa}
  \country{Israel}}
    \orcid{0000-0002-2716-086X}
\email{omerbp@technion.ac.il}

\begin{abstract}
Suppose you run a home exam, where students should report their own scores but can cheat freely. You can, if needed, call a limited number of students to class and verify their actual performance against their reported score. We consider the class of mechanisms where truthful reporting is a dominant strategy, and truthful agents are never penalized---even off-equilibrium. 

How many students do we need to verify, in expectation, if we want to minimize the bias, i.e., the difference between agents' competence and their expected grade? When perfect verification is available, we characterize the best possible tradeoff between these requirements and provide a simple parametrized mechanism that is optimal in the class for any distribution of agents' types. When verification is noisy, the task becomes much more challenging. We show how proper scoring rules can be leveraged in different ways to construct truthful mechanisms with a good (though not necessarily optimal) tradeoff. 
\end{abstract}

\keywords{Game theory; mechanism design; auditing}

\usepackage{subcaption}
\usepackage{balance} 
\usepackage{xspace}
\usepackage{enumitem}
\newenvironment{rtheorem}[1]{\medskip\noindent\textbf{Theorem~\ref{#1}. }\begin{itshape}}{\end{itshape}\medskip}
\newenvironment{rproposition}[1]{\medskip\noindent\textbf{Proposition~\ref{#1}. }\begin{itshape}}{\end{itshape}\medskip}
\newenvironment{rlemma}[1]{\medskip\noindent\textbf{Lemma~\ref{#1}. }\begin{itshape}}{\end{itshape}\medskip}
\usepackage{algorithm}
\usepackage[noend]{algorithmic}
\usepackage{amsmath}
\usepackage{amsthm}
\usepackage{amsfonts}
\usepackage{enumitem}
\usepackage{multirow}
\usepackage{xcolor}
\usepackage{mathtools}
\usepackage{tikz}
\usetikzlibrary{plotmarks}
\usepackage{pgfplots}
\usetikzlibrary{calc}
\usepackage{bm}
\usepackage{placeins}
\def\M{\mathcal{M}}

\def\HM{\mathcal{HM}}
\def\AM{\mathcal{AM}}

\def\bias{\text{bias}}
\def\ver{\text{ver}}
\def\Bias{\textsc{BIAS}}
\def\Ver{\textsc{VER}}
\def\BS{\textsc{\textbf{B}}}

\usepackage{newfloat}
\floatname{algorithm}{Mechanism}
\usepackage{listings}
\DeclareCaptionStyle{ruled}{labelfont=normalfont,labelsep=colon,strut=off} 
\floatstyle{ruled}
\newfloat{listing}{tb}{lst}{}
\floatname{listing}{Listing}

\newcount\Comments  
\newcommand{\kibitz}[2]{\ifnum\Comments=1{\color{#1}{#2}}\fi}
\newcommand{\rmr}[1]{\kibitz{red}{[RESHEF:#1]}}
\newcommand{\obp}[1]{\kibitz{blue}{[OMER:#1]}}

\newcommand{\jwr}[1]{\kibitz{orange}{[Jonny:#1]}}
\newcommand{\HR}[1]{\textbf{HR#1}}

\newcount\Full  
\newcommand{\full}[2]{\ifnum\Full=1{#1}\else{#2}\fi}

\def\set{\leftarrow}

\def\ol{\overline}
\def\eps{\varepsilon}

\def\calT{\mathcal{T}}

\newcommand{\verall}{\textsc{Verify-All}\xspace}
\newcommand{\hgp}{\textsc{Huge-Penalty}\xspace}
\newcommand{\mcv}{\textsc{MCV}\xspace}
\newcommand{\asmcv}{\textsc{AS-MCV}\xspace}
\newcommand{\payall}{\textsc{Pay-All}\xspace}

\def\bparam{\beta}
\def\lin{\theta}
\def\lv{\mathcal{LVM}}
\def\PM{\mathcal{PVM}}
\def\LV{\lv}

\def\ex{\kappa}

\begin{document}


\pagestyle{fancy}
\fancyhead{}


\maketitle 

\section{Introduction}
We consider a scenario in which self-interested agents need to report their own performance, or score in a given task, such as a home exam,  earnings, recycling, etc.

More precisely, each agent reports her own performance or competence (=expected score), and should be graded or compensated accordingly. However, agents always prefer \emph{higher grades} and thus have an incentive to misreport. 

The principal has limited verification capabilities: it can select some of the agents and observe their score directly by retaking the exam under supervision, auditing their financial status, tracking their consumption or recycling behavior, etc. The final grades can be based on both \emph{reported} and \emph{observed} scores, and thus agents that score lower than they reported could be punished. This  abstraction of verification mechanisms is common in the economic literature when the underlying problem varies from allocation of one or more goods~\cite{ben2014optimal,chua2023optimal}, to incentivizing certain actions~\cite{mookherjee1989optimal}, or revealing the value for the principal~\cite{li2020mechanism}. Intuitively, we would like a mechanism where truthful reports are a dominant strategy, as is standard both in the economics and mechanism design literature.

However, we focus on a design goal that has not been studied to the best of our knowledge, despite being very natural:  specifically, the principal would like to assign the agent(s) their true competence as their grade. This alone could be trivially obtained by verifying everyone, by threatening cheaters with an excessive penalty, or by rewarding truth-tellers. However, we would like to audit as few agents as possible, and excessive penalties may not be allowed. 

Moreover, some agents may misreport even in a truthful mechanism.  This can happen if agents are irrational~\cite{kahneman2003maps}, have unintentional mistakes in their reports~\cite{selten1988reexamination}, cannot see why the mechanism is truthful~\cite{li2017obviously, ashlagi2018stable}, cannot verify the behavior of the mechanism~\cite{branzei2015verifiably}, etc. We therefore require that a truthful agent is never punished (i.e., graded lower than her competence) regardless of others' reports.

 In many situations, there is public information on the distribution of agents' types, that the principal can use. See   Fig.~\ref{fig:FICO} and Fig.~\ref{fig:tax}. Ideally, however, the properties of the mechanism (e.g., truthfulness) should not critically depend on access to the distribution, as it may be inaccurate.



 \begin{figure}
     \centering
\includegraphics[width=0.9\linewidth]{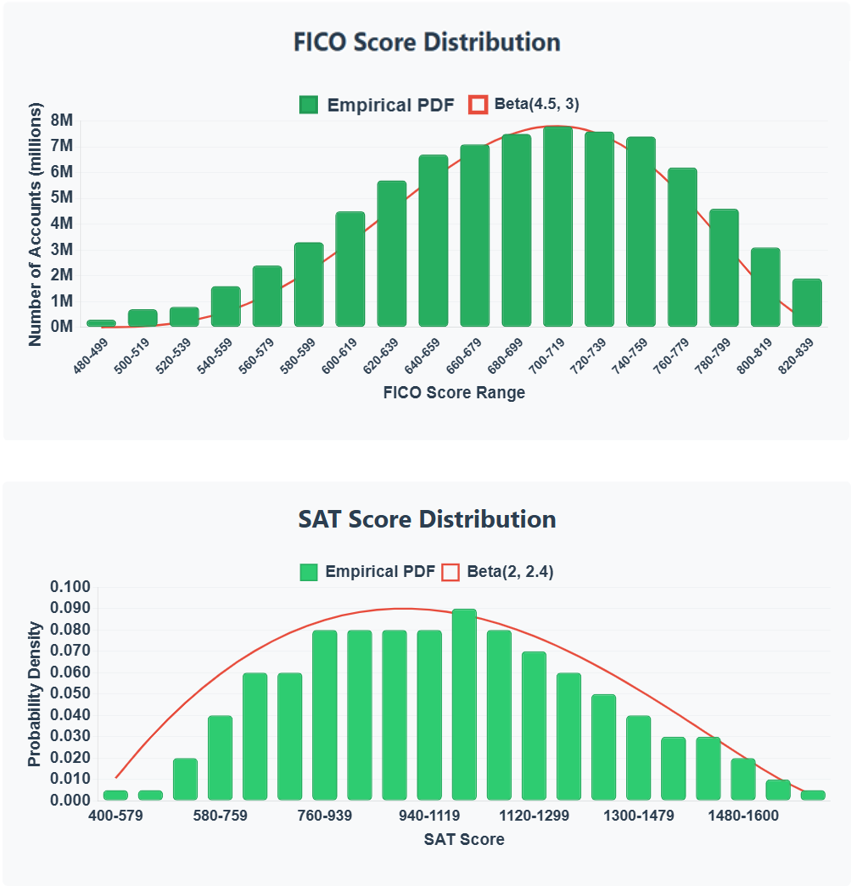}
     \caption{Top: Typical distribution of credit score~\cite{FICO}.  Bottom: Typical distribution of SAT scores~\cite{wikipediaSAT}.}
     \label{fig:FICO}
 \end{figure}

\subsection{Related Work}
\rmr{AAAI'26: Optimally Auditing Adversarial Agents}
There is much work on incentivizing self-interested agents to report private information. 
We can distinguish these settings according to the main reason agents may want to lie. 

In the first type of setting, agents do not really care how the principal uses the reported information but may lie because they do not have it or because obtaining the correct information  requires effort. Common solutions in the single agent case include the use of scoring rules~\cite{savage1971elicitation,zohar2008mechanisms} or more complicated contracts that assume some ability to verify the outcome~\cite{carroll2019robust}. In the multi-agent case, there are many variants of mechanisms built on cross-validation, such as Bayesian Truth Serum~\cite{prelec2004bayesian,witkowski2012robust}.

In contrast, the mechanism design literature that is closer to our motivation deals with situations where the agents have a direct interest in the outcome, which may be a common alternative, assignment of resources, and so on. We can further split problems into those that allow monetary transfers (which often make it easy to align incentives, despite possible computational hardness~\cite{ben2024principal}); and those that are not~\cite{nisan2007introduction}. In the latter type, negative results are abundant and successful mechanisms are only possible when the structure of the problem can be exploited. 

Our problem relates to the second type, and moreover, agents have very clear incentives to inflate their reported scores. The closest line of work within this framework is probably \emph{strategic classification}, where agents report information and would like it to be classified in a certain way~\cite{meir2012algorithms,hardt2016strategic,zrnic2021leads,levanon2021strategic}. For example, \citet{hardt2016strategic} consider a designer (`jury') that sets up the classification algorithm. An agent ('contestant') can try and game the input, at a cost, in an attempt to improve her outcome (say, increase her chances for a loan). In our model, agents can lie freely, but might be caught and punished. 

Some papers also combine payments/penalties and verification. E.g. in \cite{ceppi2019partial} partial verification is used to reduce payments; and other works consider payments/penalties contingent on some value realization or action~\cite{epitropou2019optimal,porter2008fault} (implying it can be freely verified).  \citet{mylovanov2017optimal} highlight the issue of \emph{limited liability}, as the penalty is exogenously bounded to a fraction of the value of the allocated good. 

Closer to our motivation, \citet{akasiadis2017cooperative} use scoring rules to induce truthful reports of electricity consumption shifting costs. However, in these works, all values are eventually revealed without the principal taking an explicit action. 

\rmr{see paper in folder about auditing with adversarial agents, and references therein}
\paragraph{Mechanism design with verification}
In the classical principal-agent model with verification~\cite{green1986partially}, the principal has no ability to reveal the true state, but can a-priori restrict the allowed responses of the agent, which is irrelevant in our case. 
In a recent follow-up work by Ball and Kattwinkel~\cite{ball2025probabilistic_verification}, they add \emph{tests} that the principal can assign to the agent following her report, similarly to our noisy verification model. The test is costless, and the challenge is to match each report with the best test.

Other works use costly verification in a very similar way to us, when the problem is to allocate item(s) to the agents with the highest private value~ \cite{ben2014optimal,chua2023optimal}, 
or eligibility based on private features~\cite{estornell2021incentivizing}, where value/features are sampled independently from \emph{known} individual distributions. \rmr{in our setting we can think of that as allowing a limited number of agents to pass, as grades are in $\{0,1\}$)} 
Li studies a similar assignment problem, but where the agents privately know the value of the item to the principal~\cite{li2020mechanism}. Our setting is most similar to \cite{mookherjee1989optimal}, where agents' utility is a function of their reported type and the penalties, but the principal's goals are very different. 





\subsection{Our Contribution}
In the next section, we define the class of elicitation mechanisms and our formal requirements. Section~\ref{sec:det} contains our main result, which is the single-parameter subclass of Monotone-Cutoff Verification (MCV) mechanisms. MCV provide optimal tradeoff between minimal verification and minimal bias assuming deterministic verification.  When the type distribution is known, one can choose the parameter to bound either the bias or the verification cost. When it is not known, we show how  near-optimal bounds can still be guaranteed in a multiagent setting.  In  Section~\ref{sec:noisy} we provide three mechanisms that only assume noisy verification. All mechanisms leverage \emph{proper scoring rules} albeit in different ways.  The last section shows empirically that  our mechanisms   provide a good tradeoff of bias vs. verification on common distributions. 

\full{Some proofs are deferred to the appendices to allow for continuous reading}{All missing proofs are included in the full version on arxiv~\cite{?}}

\section{Setting}
There are $N$ agents, each one with a type $t_i$, which is an element in the set $\mathcal T$, where $\mathcal{T} \subseteq [0,1]$. The type, which is private information, measures the competence of agent $i$, where better agents have a higher type. Agents report their types to a central authority that we call the \emph{principal}. We denote the report of agent $i$ by $\hat t_i\in\calT$. The principal wishes to grade agents, where the grades should be as close as possible to the real types. Agents are selfish and wish to maximize their grade, as explained in the introduction\obp{to add an example}. Thus, agents will potentially lie to maximize their expected grade.

\rmr{check out Rose et al. ’12}

\rmr{A potential addition would be a direct comparison with audit mechanisms under limited budget (e.g., Emek et al., AAMAS 2012; Daskalakis et al., EC 2020), which study related verification tradeoffs but under different incentive assumptions.}
\paragraph{Verification}
The principal can verify a particular agent's type, e.g., by forcing them to re-take an exam in class,  audit economic reports, or track individual consumption. We assume the principal has access to a \emph{verification process}, which returns $t_i$ in the deterministic case. In the noisy verification case, the process returns  $s_i\sim T_i$ where $T_i$ in an arbitrary distribution on $\calT$ with mean $t_i$. Other than the expectation, we make no assumptions on the distribution $T_i$. In particular, it may be correlated with other agents, and need not be known to the designer. 




\paragraph{Mechanisms}
A mechanism $\M=(\M_V,\M_G)$ consists of two different functions:
\begin{enumerate}
    \item Verification selection function $\M_V:\calT\rightarrow[0,1]$. This function takes as input an agent's reported types $\hat t_j$ and possibly other available information (see below), and returns an \emph{audit probability} $q_i\in[0,1]$ for every $i \in N$: the probability that $i$ will be selected for verification. Since all agents with same reported type $t$ have the same $q_i$, we also denote $q_t$;
    \item Grading function $\M_G:\calT\times(\calT\cup\{\bot\})\rightarrow \mathbb R$. This function takes as input a reported type $\hat t_i$ and a realized performance sample $s_i$ (or $\bot$ if $i$ is not verified) and sets a grade $g_i$. \rmr{We could, in principle, also rely on information from other agents in setting $i$'s grade, but this is not used in the current paper. It is allowed in \cite{chua2023optimal} as they let the allocation depend on all reports and all verification outcomes. Their mechanism only needs to set an allocation probability for each non-verified agent.} 
  The grade is not necessarily in the set $\calT$ and could e.g. be negative unless specified otherwise. \jwr{A grade $g_i < \inf (\calT)$  or $g_i > \sup (\calT)$ corresponds to  an excess penalty or reward, respectively. I see a problem here..}   
\end{enumerate}
The functions above may rely on external parameters or additional information, that the mechanism takes as input. 

\begin{algorithm}[t!]
\caption{A general elicitation mechanism}
\label{alg:algorithm}
\begin{algorithmic}[1] 
\REQUIRE $(\hat t_i)_{i\in N}, \M_V, \M_G$,[Optional: Information $I$]
\FORALL{$i\in N$}
\STATE $q_i  \gets \M_V(\hat t_i)$ \hfill \textit{//audit probability}
\ENDFOR
\FORALL{$i\in N$}
\STATE with probability $1-q_i$, set  $g_i\leftarrow \M_G(\hat t_i,\bot)$;
\STATE with the remaining probability, obtain a verified score $s_i\sim T_i$ and set grade $g_i\leftarrow \M_G(\hat t_i,s_i)$.
\ENDFOR
\STATE \textbf{return} $(g_i)_{i\in N}$
\end{algorithmic}
\end{algorithm}

To specify a concrete mechanism, we only need to formally define $\M_V$ and $\M_G$ (see Mechanism~\ref{alg:algorithm}). Three simple examples are:
\begin{description}
    \item[$\verall$:] $\M^{VA}_V(\hat t_i)=1\ \forall \hat t_i$; and grades such that\\
    $\M^{VA}_G(\hat t_i,s_i) = \hat t_i$ if $\hat t_i=s_i$, and $0$ otherwise.
    \item[$\payall$:] $\M^{PA}$ just pays 1 to all agents without verification, i.e., $\M^{PA}_V\equiv 0,\M^{PA}_G\equiv 1$. 
    \item[$\hgp$:]$\M^{HP}_V(\hat t_i)= \eps\ \forall \hat t_i$  for some small  $\eps>0$; and grades each  such that  
        \(
    \M^{HP}_G(\hat t_i,s_i) := 
    \begin{cases}
        \hat t_i & \text{if } \hat t_i=s_i \text{ (or $\bot$)}\\
        -2/\eps & \text{if }\hat t_i \neq s_i \\
    \end{cases}
    \)
\end{description}
\newcommand{\hatb}[1]{\hat{\bm{#1}}}
Note that once we fix the mechanism $\M$ and the profile $\hatb t$, the expected grade of an agent can be determined. 
Let
\begin{equation}
    \label{eq:g}
    \bar g_t(\hat t) := \M_V(\hat t) E_{s_i\sim T_i}[\M_G(\hat t_i,s_i)]+(1-\M_V(\hat t))\M_G(\hat t_i,\bot) 
\end{equation}
be the expected grade of a type~$t$ agent reporting $\hat t$, and denote by $\bar g_t:=\bar g_t(t)$ the expected grade of a truthful type-$t$ agent.

\paragraph{Hard requirements}

\begin{enumerate}[leftmargin=0.4cm,itemindent=.75cm,label=\textbf{HR\arabic*}.]
    \item Truth-telling must be a weakly dominant-strategy; \rmr{define?}
    \item Truth-tellers are never punished. I.e.,  $g_{i} \geq t_i$ for all $t_i\in \calT$ even if others lie;
      \item Realized grades must be at least $-\xi$, where $\xi\geq 0$ is a parameter.\footnote{We could similarly add a requirement about the \emph{maximal} realized grade (e.g. a course grade may be capped at `100' by the system). For simplicity we avoid this restriction, but note that most of our mechanisms indeed keep realized grades at most 1.} 
\end{enumerate}
The special case of $\xi=0$ is called \emph{limited liability} and means we cannot `punish' cheaters, beyond disqualifying their reward.

In the \emph{deterministic verification} case, a mechanism satisfying \HR1-\HR3 is called \emph{valid} (we later discuss relaxations for the noisy setting). Clearly if we only care about validity, it could be obtained trivially, by either one of the $\verall$ or $\payall$ mechanisms. $\hgp$ may or may not be valid, depending on $\xi$.

\paragraph{Soft requirements and tradeoffs}
Ideally, we would like to minimize both the amount of verified agents, and the grading error $|g_i-t_i|$. Since the grade may be affected by random factors, a more modest goal would be to minimize the \emph{bias} $|\bar g_t-t|$ of every type $t$. By \HR2 the bias of a valid mechanism must be non-negative.

The metrics we would like to optimize may also depend on the type distribution $p=(p(t))_{t\in \calT}$,\footnote{For ease of presentation we treat $p$ as a discrete distribution but all definitions and results (except in Section~\ref{sec:det_hist}) naturally extend to continuous distributions. } where $t_{min}\neq t_{max}$ are the minimal and maximal types with positive support. 
All soft requirements refer to a truthful profile:
\begin{itemize}
    \item Minimize the expected proportion of agents we verify \\$\ver(M,p):= E_{t\sim p}[q_t] =\sum_{t\in \calT}p(t)q_t$.
    \item Minimize the  \textbf{maximal}  bias (always nonnegative by \HR2)\\$\BS(M):=\max_{t\in \calT} (\bar g_t-t)$.
    \item Minimize the  \textbf{expected} bias\\ $\bias(M,p):=E_{t\sim p}[\bar g_t-t]=\sum_{t\in \calT} p(t)(\bar g_t-t)$.
\end{itemize}
While the soft requirements necessarily depend on the type distribution $p$, our hard requirements including truthfulness should hold even if $p$ is unknown (in contrast to most previous work~\cite{ben2014optimal,estornell2021incentivizing}).

\section{Deterministic Verification}\label{sec:det}
In this section we assume verification is deterministic, i.e., whenever $i$ is audited, we get $s_i=t_i$ w.p.~1. 
Even so, the three soft requirements capture design goals that are generically contradictory:
the lower the audit probability, the higher the required threat to maintain truth-telling as an equilibrium. The $\verall$ and  $\hgp$ mechanisms represent the two extremes. 
It is thus natural to ask what is the efficient compromise between accuracy and small sample, in terms of a theoretical bound and, moreover, whether we can find mechanisms that implement it. 

\begin{definition}[Domination]
We say that a mechanism $\M$ (weakly) \emph{dominates} some other mechanism $\M'$ if for any type distribution $p$, $\M$ outperforms $\M'$ in all soft requirements. If in addition $\M$ strictly outperforms $\M'$ on one of those requirements, we say that \emph{strictly} dominates $\M'$.  
\end{definition}
In other words, $\M$ dominates $\M'$ if $\M$ verifies fewer agents, and yet obtains lower maximal bias and expected bias. Furthermore,
\begin{definition}[Efficiency]
     A mechanism is \emph{efficient} if there is no other mechanism that strictly dominates it.     
 \end{definition}

\subsection{The Monotone-cutoff Verification Class} 
\begin{figure}
    \centering
    \newcommand{\gammaval}{2.}
\newcommand{\gammafix}{0.02}
\newcommand{\xival}{0.75}
\newcommand{\maxt}{7.}
\begin{tikzpicture}[yscale=0.7]
    \draw[->] (0,0) -- (\maxt-.5,0) node[right] {$t$};
    \draw[->] (0,-1) -- (0,\maxt-.5) node[above] {};

    \draw (\gammaval, -0.15) -- (\gammaval, 0.15);
     \draw (\maxt-1, -0.15) -- (\maxt-1, 0.15);
   
    \draw (-0.15, \gammaval) -- (0.15, \gammaval);
        \draw (-0.15, 0) -- (0.15, 0);
    \draw (-0.15, -\xival) -- (0.15, -\xival);
    
    \node[below] at (\gammaval, -0.15) {$\gamma$};
        \node[below] at (\maxt-1, -0.15) {$1$};
    \node[left] at (-0.15, \gammaval) {$\gamma$};
     \node[left] at (-0.15, 0) {$0$};   
    \node[left] at (-0.15, -\xival) {-$\xi$};

    \draw[dotted, thin] (0,0) -- (\maxt-0.8,\maxt-0.8);

        \draw[thick,dashed,orange] (\gammaval,\gammaval+2*\gammafix) -- (\maxt-1,\maxt-1+2*\gammafix);
        
    \draw[thick,blue,dash dot] (\gammaval,-\xival) -- (\maxt-1,-\xival);

\draw[thick, green!60!black] (0,\gammaval-2*\gammafix) -- (\gammaval,\gammaval-2*\gammafix);
\draw[thick, green!60!black] (\gammaval,\gammaval-2*\gammafix) -- (\maxt-1,\maxt-1-2*\gammafix); 

    \fill[gray, opacity=0.2] 
    (0,0) -- (0,\gammaval) -- (\gammaval,\gammaval) -- cycle;

\draw[thin, double, red] (0,0) -- (\gammaval,0);
\draw[thin, double, red, domain=\gammaval:\maxt-1, smooth, variable=\t] 
    plot ({\t}, {(\t - \gammaval)*\maxt/(\t + \xival)});

    \fill[red!20, opacity=0.4, domain=\gammaval:\maxt-1, smooth, variable=\t] 
    (0,0) -- (\gammaval,0) -- 
    plot ({\t}, {(\t - \gammaval)*\maxt/(\t + \xival)}) -- 
    (\maxt-1,0) -- cycle;

\begin{scope}[shift={(0.5,6)}] 
    \draw[thick] (-0.3, -2.4) rectangle (2.6, 0.5); 

    \draw[thick, green!60!black] (0,0) -- (0.6,0);
    \node[right] at (0.7,0) {$g_i$ for $\bot$}; 

    
    \draw[dashed, thick, orange] (0,-0.5) -- (0.6,-0.5);
    \node[right] at (0.7,-0.5) {$g_i$ for $s_i=\hat t_i$};
    
    \draw[thick, blue, dash dot]
     (0,-1) -- (0.6,-1);
    \node[right] at (0.7,-1) {$g_i$ for $s_i \neq \hat t_i$}; 

        \draw[thick, double, red]
     (0,-1.5) -- (0.6,-1.5);
    \node[right] at (0.7,-1.5)  {$q_i$}; 

  \draw[thick,dotted]
     (0,-2) -- (0.6,-2);
    \node[right] at (0.7,-2)  {$t$}; 

\end{scope}
\end{tikzpicture}
    \caption{Illustration of an MCV mechanism.  The horizontal line represents the report $\hat t_i=t$, and the vertical axis plays the role of both the grade and the audit probabilities. The audit probability (red double line) is zero for reports below $\gamma$,  and any such report gets $\gamma$ (green horizontal segment). The dotted diagonal line represents the identity. When $\hat t_i>\gamma$, audit probability is increasing, both the default (green) and truthful (dashed orange) grades match $\hat t_i$, and there is a fixed maximal penalty of $-\xi$ (blue dash-dots).}
    \label{fig:MCV}
\end{figure}
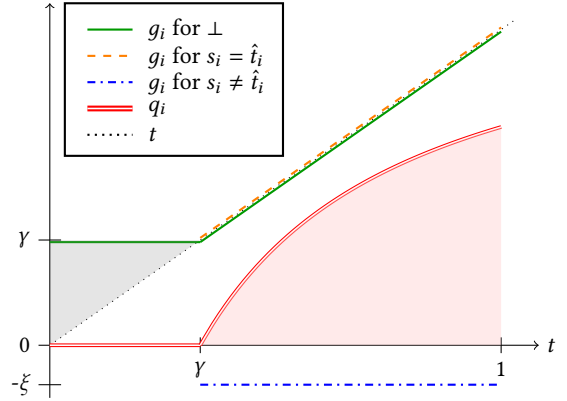
Next, we introduce the class of \emph{Monotone-Cutoff Verification} mechanism (hereinafter \mcv). The class is parameterized with $\gamma \in [0,1]$, which governs the audit probabilities. We emphasize that for an \mcv mechanism and the results in this subsection, the number of agents is irrelevant, and we could have stated them for a single agent. However, our extensions will make use of the entire population.

\paragraph{\underline{The Monotone-Cutoff Verification (\mcv) mechanism:}} The mechanism $\M^{\gamma,\xi}$ takes the parameters $\gamma\in[0,1]$ and $\xi\geq 0$ as input. Since $\xi$ is fixed and derived from \HR3 we will omit it to simplify the notation, writing $\M^{\gamma}$ rather than $\M^{\gamma,\xi}$.

\begin{itemize}
    \item Verification: agent $i$ with report $\hat t_i$ is verified by $\M^{\gamma}_V$ with a probability of 
    \(
    \M^{\gamma}_V(\hat t_i)=\begin{cases}
        \frac{\hat t_i-\gamma}{\hat t_i+\xi}\quad &  \hat t_i > \gamma\\
        0\quad & \hat t_i \leq \gamma
    \end{cases}
    \)
    \item Grading: if agent $i$ was not verified, then she receives a score of 
    \(
    \M^{\gamma}_G(\hat t_i,\bot)=\begin{cases}
        \hat t_i \quad &  \hat t_i > \gamma\\
        \gamma \quad & \hat t_i \leq \gamma
    \end{cases}.
    \)\\
    Otherwise, if the verification yielded a result $s_i$, then
    
    \(\M^{\gamma}_G(\hat t_i,s_i)=\begin{cases}
        \hat t_i \quad &  \hat t_i =s_i\\
        -\xi \quad &  \hat t_i \neq s_i
    \end{cases}.
    \)
    \end{itemize}     

We visualize an example of a MCV mechanism in Figure~\ref{fig:MCV}. The red and gray shaded areas demonstrate the overall audit probability and bias, respectively (for a uniform type distribution). 

Note that for $\xi=0$ and the extreme values $\gamma=0$ and $\gamma=1$, we get the $\verall$ and $\payall$ mechanisms, respectively. Setting $\gamma=0, \xi=\frac1\eps$ almost coincides with the \hgp mechanism. Thus the \mcv class of mechanisms aims to optimally balance between the solutions, where every pick of $\gamma,\xi$ corresponds to a different point on the Pareto frontier.  

Notice that $\M^{\gamma}$ effectively lets agents report their type as ``not greater than $\gamma$'' without being verified, and $\bar g_t=\max\{t,\gamma\}$. 
\begin{proposition}
    All \mcv mechanisms are valid.
\end{proposition}

\begin{proof}
    \textbf{HR2} and \textbf{HR3} are immediate by definition. For truthfulness, note that in $\M^\gamma$ truth telling guaranties at least $\gamma$ for all agents, whereas the maximum expected payoff for untruthful agent,
    \[\max_{\hat t_i \neq t_i}\big(1-\M^\gamma_V(\hat t_i)\big)\M^\gamma_G(\hat t_i,\bot)+\M^\gamma_V(\hat t_i)\cdot \M^\gamma_G(\hat t_i, t_i),\]
    is exactly $\gamma$.  If $\hat t_i \leq \gamma$, the agent receives a score of $\gamma$, without being verified. For $\hat t_i > \gamma$: 
    \begin{align*}
        &\big(1-\M^\gamma_V(\hat t_i)\big)\M^\gamma_G(\hat t_i,\bot)+\M^\gamma_V(\hat t_i)\cdot \M^\gamma_G(\hat t_i , t_i)\\
       = &\Big(1-\frac{\hat t_i -\gamma}{\hat t_i +\xi}\Big)\hat t_i +\frac{\hat t_i -\gamma}{\hat t_i +\xi}(-\xi)= \gamma. \qedhere 
    \end{align*}   
\end{proof}

It is easy to see that setting $\gamma=0$ and $\gamma=1$ (under limited liability) results in the $\verall$ and $\payall$ mechanisms, respectively. 
Next, we show that the class $\Gamma:=\{\M^{\gamma}\ :\ \gamma \in [0,1]\}$ constitutes the efficiency boundary of all valid mechanisms. 

\begin{theorem}\label{thm:optimal_class}
     For any valid mechanism $M$, there exists $\M^\gamma \in  \Gamma$ that weakly dominates $M$. 
\end{theorem}
Meaning, if minimum possible grade is $-\xi$ and we demand \textbf{HR1} and \textbf{HR2}, we only need to choose $\gamma \in [0,1]$ for which the corresponding $\M^\gamma \in  \Gamma$ best suits our ideal trade-off between accuracy and verification (where these depend also on the prior distribution of types $p$). In particular, as the proof shows, $\ver(\M^\gamma,p),\BS(\M^\gamma)$ and $\bias(\M^\gamma,p)$ are easy to compute for any given type dist. $p$.  
\begin{proof}
  Denote, for any type distribution $p$ and $\gamma\in [0,1]$ the following measures that are independent of the mechanism:
\begin{align}
    \Bias(\gamma,p)&:=\sum_{t\leq \gamma}p(t)(\gamma-t);\label{eq:err} \\
    \Ver(\gamma,p)&:= \sum_{t>\gamma}p(t)\frac{t-\gamma}{t+\xi} \label{eq:ver}
\end{align}

We first argue that any mechanism $\M^\gamma$ meets these measures exactly, i.e., $\forall \gamma\in[0,1]$ and $p\in \Delta([0,1])$, $\bias(\M^\gamma,p)=\Bias(\gamma,p)$ and $\ver(\M^\gamma,p)=\Ver(\gamma,p)$.
 
  Indeed, these equalities follow directly from the definitions of $\M^\gamma_G$ and $\M^\gamma_V$, respectively. Also note that $\BS(\M^\gamma)=\gamma-t_{min}$.
  

  Now take $M$ as described in the theorem. Let $M'$ be a mechanism identical to $M$, except
  $M'_G(\hat t_i, s_i):=-\xi$ for $\hat t_i\neq s_i$. We argue that $M'$ is also valid, and dominates $M$. First, since penalty for cheating is weakly increasing by \HR3, being truthful remains a weakly dominant strategy. Second, grades of all truthful agents are unchanged, and thus no change in bias.  
  Next, let
  \begin{align*}
      \tilde \gamma &:= \max_{\hat t_i \neq t_i}\big(1-M'_V(\hat t_i)\big)M'_G(\hat t_i,\bot)+M'_V(\hat t_i)\cdot M'_G(\hat t_i, t_i)\\
      &= \max_{\hat t_i}\big(1-M'_V(\hat t_i)\big)M'_G(\hat t_i,\bot)+M'_V(\hat t_i)\cdot (-\xi).
  \end{align*}
 $\tilde\gamma$ the maximum gain for untruthful agents under $M'$, regardless of their true type. We will see that $\M^{\tilde \gamma}\in \Gamma$ dominates $M'$ (and thus $M$).
 
 If a grade of $\tilde \gamma$ is achievable by lying, then by \textbf{HR1} all truthful agents are graded at least $\tilde \gamma$ in equilibrium, therefore 
 \[\bias(M',p) \geq \sum_{t \leq \tilde \gamma}p(t)(\tilde \gamma-t)=\bias(\M^{\tilde \gamma})\]
and, in particular, a type-zero agent has at least $g_0=\tilde \gamma$, thus $\BS(M') \geq \tilde \gamma -t_{min}=\BS(\M^{\tilde \gamma})$. 
Moreover, for all $\hat t_i \in [0,1]$ we have $M'_G(\hat t_i,\bot)\geq \hat t_i$ by \HR2, and thus by definition of $\tilde \gamma$,  
\begin{align*}
   \tilde \gamma 
   & \geq \big(1-M'_V(\hat t_i)\big)\hat t_i - M'_V(\hat t_i)\xi
   \implies M'_V(\hat t_i)  \geq \frac{\hat t_i - \tilde \gamma}{\hat t_i +\xi} 
\end{align*}
 Summing over all types $ \tilde \gamma  < \hat t_i \leq 1 $ we get 
  \[\ver(M') \geq \sum_{t > \tilde \gamma}p(t)\frac{t-\tilde \gamma}{t+\xi}=\ver(\M^{\tilde \gamma}).\qedhere\]
\end{proof}

\begin{corollary}
    Every \mcv mechanism is efficient.
\end{corollary}
\begin{proof} If $\M^{\gamma}$ is strictly dominated by any other mechanism, the theorem gives that it is also strictly dominated by some $\M^{\gamma'} \in \Gamma$. But this is impossible since for any distribution $p$ and for any two $0 \leq \gamma <t^*< \gamma' \leq 1$ (where $p(t^*)>0$), $\ver(\M^{\gamma'},p) < \ver(\M^{\gamma},p)$ while $\bias(\M^{\gamma'},p) > \bias(\M^{\gamma},p)$.
\end{proof}

\paragraph{Optimizing the cutoff point}
Theorem~\ref{thm:optimal_class} provides us with a way to explicitly bound the tradeoff on bias/verification.

For  $\bparam\in [0,1]$ let 
$\gamma(\bparam,p):= \max\{\gamma\in [0,1] : \Bias(\gamma,p)\leq \bparam\}.$

\begin{corollary} \label{cor:beta_bound}
  For every $\bparam \in [0,1]$,  type dist. $p$,\\  $\bias(\M^{\gamma(\bparam,p)},p)=\bparam$; and $\ver(\M^{\gamma(\bparam,p)},p)=\Ver(\gamma(\bparam,p),p)$.
\end{corollary}
 As we saw in the proof of the theorem, no mechanism with lower bias can verify less.

\paragraph{Limited liability}
In particular, if $\xi=0$ the following corollary is implied, showing that \textit{some} bias is inescapable unless we verify each and every agent.
\begin{corollary}\label{cor:LL err/smpl}
    Suppose $p(0)>0$. Under limited liability, no valid mechanism exists such that $\bias(M,p)=0$ and $\ver(M,p)<1$. 
\end{corollary}
\begin{figure}
    \centering
    \includegraphics[width=0.9\linewidth]{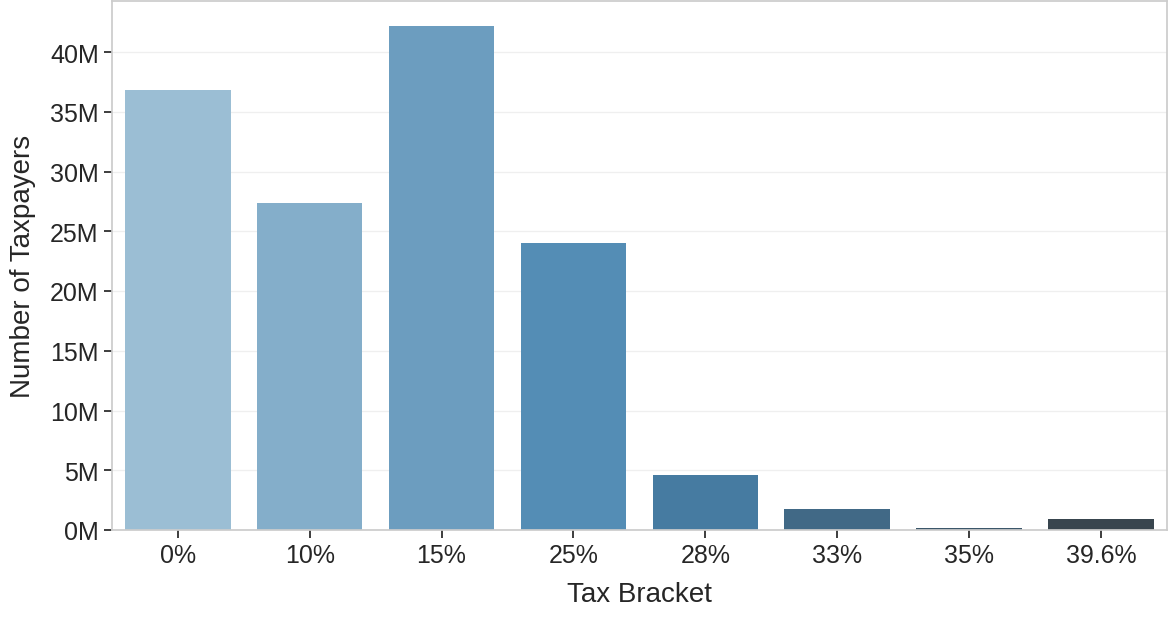}
    \caption{ US income distribution by tax brackets in 2013~\cite{IRS}. }
    \label{fig:tax}
\end{figure}
\paragraph{Example}
 Suppose we apply the MCV mechanism to decide on tax audits on the population from Fig.~\ref{fig:tax}, where a grade of `1' means full tax exemption (0\% tax). Note that there are no `grade 0' (100\% tax) agents so Corollary~\ref{cor:LL err/smpl} does not apply. In fact by  setting $\gamma=0.6$ and $\xi=0$ (i.e., tax full income of cheaters) we can get 0 bias by auditing  30\% of the population. 
 
 In contrast, if we maintain limited liability and set $\gamma=0.75$ (cap  max tax bracket at 25\%) then we will have an expected bias of $0.003$, auditing only 14.6\% of the people.
 
\subsection{Unknown Prior}\label{subsec:unknown prior}
The definition of the \mcv class, its strategy-proofness and efficiency do not require an a-priori assumed type distribution. However, running $\M^\gamma$ for some arbitrary $\gamma \in [0,1]$ offers no guarantee in terms of average bias and expected sample size. 

In what comes next, we show that we can in fact approximate pretty well any desired goals on the Pareto curve, even without a prior, provided there are enough agents from the same distribution $p$. Notably, the agents' types need not be independent.

The basic principle of the  mechanism we construct is that for every agent $i$, we pick a slightly different \mcv mechanism,  \emph{based on the observed distribution of all agents excluding her}. As we will see, this scheme will approximate the desired bias goal by an additive $O(1/n)$ factor, with no additional verification. 

\begin{definition}
    For all $i \in N$, let $p_{-i}: [0,1] \to [0,1]$ be the observed  (reported) type distribution aggregated over $\hat t_j, j\neq i$. Since $p_{-i}$ is identical for every two agents of the same type, we also use the notation $p_t$. 
\end{definition}

\paragraph{\underline{Agent-Specific \mcv (\asmcv) mechanism:}}
For all $\bparam\in [0,1]$, define   $\AM^{\bparam}$ as follows.
For every agent $i$, let $\gamma_i:= \gamma(\frac{n}{n-1}\bparam,p_{-i})$; 
then apply $\M^{\gamma_i}$ to agent $i$.

\medskip
We first observe that $\AM^\bparam$ is valid, since any $\M^{\gamma_i}$ is valid, and since the report of $i$ cannot influence which mechanism she will face. 
 $\bparam$ is our desired bound on the expected bias, and we wish to minimize $\ver(M,p)$ \emph{without knowing $p$}.

 \begin{theorem}\label{thm:AS-MCV efficiency}
    For every $\bparam \in [0,1]$, and any type dist. $p$,
    \begin{enumerate}
        \item $\bias(\AM^{\bparam},p) \leq \bparam + \frac{1}{n}$
        \item $\ver(\AM^{\bparam},p) \leq \Ver(\gamma(\bparam,p),p)$.
    \end{enumerate}
   
 \end{theorem}
 Compare these bounds to Cor.~\ref{cor:beta_bound}: 
the theorem shows that we  only lose an additive factor of $\frac1n$ on the average bias by not knowing $p$ in advance, without additional verification. 
 
 \full{The proof is in the appendix.}{} Intuitively, the challenge is that there is no longer a single cutoff. But the maximal and minimal $\gamma_{-i}$ can be used to consider all relevant types for $\bias$ and $\ver$, respectively.  Since $p_{-i}$ is close to $p$ the increase in bias is bounded. Then we show that all agents that affect $\ver()$ use the same mechanism $\M^{\gamma(\bparam,p)}$.
 
 It is not hard to see how one can similarly minimize the bias subject to bounded verification capacity.

\section{Noisy Verification}\label{sec:noisy}
The MCV mechanism depends on the fact that verification is deterministic, i.e., that liars are always caught, if verified. In this section we relax the assumption of deterministic verification, meaning that $T_i$ is not guaranteed to return $t_i$, but rather some random variable $s_i \sim T_i$ with expected value $t_i$. 

To see why MCV (and any mechanism with  penalties) fails under noisy verification, note that even a truthful agent might get $s_i<\hat t_i$ and be punished, in violation of \HR2. We therefore  have to relax \HR2 to \HR2': we still require unverified agents to get at least $\hat t$, and truthful agents should get at least $\hat t$ only \emph{in expectation} (i.e., $\bar g_t\geq t$). In other words, we require that the bias is nonnegative. 


All mechanisms we construct in this section make use of \emph{proper scoring rules}.
\paragraph{Scoring rules}
Based on \cite{gneiting2007strictly}, a \emph{proper scoring rule} is a function $R:[0,1]^2\rightarrow \mathbb R$ s.t. $E_{s\sim T}[R(x,s)]$ is maximized for $x= E_{s\sim T}[s]$ for any distribution $T$.\footnote{More generally, scoring rules allow the agent to report her full distribution $T$ rather than just its expectation, but this simplified form will suffice for our needs.}

A common example of a proper scoring rule is the \emph{continuous semi power
scoring rule} 
\begin{equation}
    R^\alpha(x,s):=\alpha \cdot s\cdot x^{\alpha-1}-(\alpha-1)x^\alpha. \label{eq:R_alpha}
\end{equation}
This is an asymmetric variant of the power scoring rules studied in \cite{selten1998axiomatic}. For $\alpha=2$ we get the semi-quadratic scoring rule $R^2(x,s)=2sx-x^2$.  It is known from \cite{gneiting2007strictly} that a function can be implemented by a proper scoring rule if and only if it is strictly convex, and this will be useful for us later on. In particular, $E_{s\sim T}[R^\alpha(t,s)]=t^\alpha$ for any distribution $T$ with mean $t\in[0,1]$, and any $\alpha>1$. 

\paragraph{Na\"ive mechanisms}
In the deterministic verification case, there were simple valid mechanisms that minimized \emph{either} the bias \emph{or} the audit probability (while performing poorly on the other). Indeed, the $\payall$ mechanism still applies in the noisy setting. A parallel of $\verall$ is also possible, by applying the scoring rule $R^{1+\eps}(\hat t,s)$ to \emph{all} agents, for some arbitrarily small $\eps>0$.\footnote{This introduces a small \emph{negative} bias, which is not allowed, but we can fix this by adding a small constant to the grade, see below.} 

However, designing a mechanism with nontrivial bounds both on the bias and the verification cost is much more challenging than in the deterministic case, since we can never be certain that a particular agent has cheated. E.g., it is not a-priori clear how to implement a mechanism similar to  $\hgp$ (even ignoring \HR3), as in the noisy setting large penalties may violate both \HR1 and \HR2'. We will later return to this point. We also show \full{in Appendix~\ref{apx:no_score}}{in the full version} that applying a proper scoring rule only to some agents is \emph{not} truthful.


We therefore provide two partial solutions in this section. The first maintains all of our previous assumptions, implementing mechanisms that behave like `scoring-rules-in-expectation'. The second mechanism exploits the size of the population and assumes access to the accurate distribution of the true types, and has a unique equilibrium with no verification and  arbitrarily small bias.

\paragraph{Deliberate failures}
One may worry about situations where agents deliberately under-perform when getting verified. We highlight that although we made no explicit requirement regarding this type of manipulation, the grade in all of our mechanisms is monotonic in the realization $s$, and thus agents cannot gain from such behavior.

\subsection{Scoring Rules `in Expectation'}

The idea is not to apply a scoring rule only on verified agents, but to set up the mechanism such that every agent is subject to a scoring rule \emph{in expectation}, and thus has an incentive to be truthful.

We demonstrate this first with a simple, parameter-free mechanism with linear audit probability. 
\paragraph{\underline{Linear Verification (LV) mechanism:}}

\begin{itemize}
    \item $\lv_V(\hat t_i) := \hat t_i$.
    \item $\lv_G(\hat t_i,\bot) := \frac{1}{4} +\hat t_i$.
    \item $\lv_G(\hat t_i, s_i) := 2s_i - \frac34$.
\end{itemize}

\begin{lemma}\label{lemma:LV_g}
    $\bar g_t(\hat t) = \frac{1}{4} + R^2(\hat t,  t)$, where grade expectation is taken over both the random verification decision of the mechanism, and the performance of the verified agent.
\end{lemma}

    \begin{corollary}
        The LV mechanism satisfies \HR1,   \HR2', and \HR3 with $\xi=\frac{3}{4}$.
    \end{corollary}
    \begin{proof}
        Truthfulness (\HR1) follows immediately from Lemma~\ref{lemma:LV_g} since the expected grade is an affine transformation of $R^2$, and thus also a proper scoring rule.  
        As for \HR2', note that the expected grade of a truthful agent of type $t$ is $\frac{1}{4}+t^2$. It is easy to check that this is at least $t$.  \HR3 is immediate from the definition.
     \end{proof}
 Note that the constant boost of the grade cannot be lower than $\frac1{4}$, since for $t^*=\frac{1}{2}$ the expected grade is exactly $\frac12$.

     As for the soft requirements, it is immediate that $\BS(\lv)=\frac{1}{4}$. The expected verification [resp., bias] depend on the distribution $p$, but are easy to compute by integrating $t$ [resp., $\frac14+t^2-t$] over $p$. E.g. for a uniform type distribution we would get $\ver(\LV,p)=\frac12$ and $\bias(\LV,p)=\frac1{12}$.

    \paragraph{Reducing the bias}
    While increasing the parameter $\lin$ from 0 towards 1 clearly increases the amount of verification, we saw that the effect on the expected bias is non-monotone. Increasing it beyond~1 is unlikely to reduce the bias and therefore we need a different mechanism. 
    
     \paragraph{\underline{Polynomial Verification (PV) mechanism:}}
     The mechanism takes two parameters:
     \begin{itemize}
         \item an integer parameter $\ex\geq 1$;
         \item a real-valued parameter $\lin\in[0,1]$;
     \end{itemize} 
    and is defined as follows. 
\begin{itemize}
   \item $\PM^{\lin,\ex}_V(\hat t_i) :=  \lin \cdot(\hat  t_i)^{\frac{1}{\ex}}$.
    \item $\PM^{\lin,\ex}_G(\hat t_i,\bot) := \frac{\ex^\ex}{(\ex+1)^{\ex+1}}+\frac{1}{\ex\cdot \lin^{\ex+1}}\sum_{\ell=1}^\ex (\lin \cdot \hat t_i)^{\ell/\ex}$.
    \item $\PM^{\lin,\ex}_G(\hat t_i, s_i) :=\frac{\ex^\ex}{(\ex+1)^{\ex+1}}+\frac{1}{\lin} (1+\frac{1}{\ex})s_i -\frac1{\ex\cdot\lin^{\ex+1}}$.
\end{itemize}
Note that for $\ex=1,\lin=1$ the PV mechanism coincides with the LV mechanism.\footnote{We can also see that when $\lin=1,\ex\rightarrow\infty$ the mechanism approaches $\verall$. We get  $\payall$ when $\lin=\ex=0$ although we did not formally allow $\ex=0$.} Just like the LV mechanism, any PV mechanism implements a (semi-polynomial) scoring rule. See the dash-dotted purple line in Fig.~\ref{fig:PV}.

\begin{lemma}\label{lemma:PV_g}
    $\bar g_t(\hat t) = \frac{\ex^\ex}{(\ex+1)^{\ex+1}}+  R^{1+\frac{1}{\ex}}(\hat t,  t)$;
\end{lemma}

\begin{lemma}\label{lemma:ex_bound}
    $\frac{\ex^\ex}{(\ex+1)^{\ex+1}}+t^{1+\frac1\ex} \in [t,t+\frac1{\ex}]$.
\end{lemma}
\begin{proof}
    Clearly, $\frac{\ex^\ex}{(\ex+1)^{\ex+1}}$ is positive and upper bounded by $\frac{1}{\ex+1}$.
    Now consider $t-t^{1+\frac1\ex}$. This term is nonnegative and concave. By derivation, the maximum is obtained at $(\frac{\ex}{\ex+1})^{\ex}$. Thus
    \begin{align*}
        t-t^{1+\frac1\ex} &\leq (\frac{\ex}{\ex+1})^{\ex}- (\frac{\ex}{\ex+1})^{\ex(1+\frac1\ex)}\\
        &=(\frac{\ex}{\ex+1})^{\ex}- (\frac{\ex}{\ex+1})^{\ex+1} = \frac{\ex^{\ex}(\ex+1)}{(\ex+1)^{\ex+1}}-\frac{\ex^{\ex+1}}{(\ex+1)^{\ex+1}}\\
        &=\frac{\ex^{\ex+1}+ \ex^{\ex}-\ex^{\ex+1}}{(\ex+1)^{\ex+1}} = \frac{ \ex^{\ex}}{(\ex+1)^{\ex+1}}.\qedhere
    \end{align*} 
\end{proof}
 \begin{proposition}
        The PV mechanism satisfies \HR1 and   \HR2'. It satisfies \HR3 iff       $\frac1{\ex\cdot \lin^{\ex+1}}-  \frac{\ex^\ex}{(\ex+1)^{\ex+1}} \leq \xi$.
    \end{proposition}
    \begin{proof}
        Truthfulness (\HR1) follows immediately from Lemma~\ref{lemma:PV_g} since the expected grade is an affine transformation of $R^{1+\frac1\ex}$, and thus also a proper scoring rule.  
        As for \HR2', note that the expected grade of a truthful agent of type $t$ is $\frac{\ex^\ex}{(\ex+1)^{\ex+1}}+t^{1+\frac1\ex}$. By Lemma~\ref{lemma:ex_bound} this is at least $t$.
        \HR3 is respected by our restriction on the parameters.
     \end{proof}

As for the soft requirements, it is easy to see that the audit probability is monotonically increasing in both parameters. As $\ex\rightarrow \infty$, it approaches a constant function that verifies every agent (unless $\hat t_i=0)$ w.p. $\lin$.

We observe that the bias does not depend on $\lin$ at all, which may at first seem surprising. 
The bias for a truthful type-$t$ agent is  $\frac{\ex^\ex}{(\ex+1)^{\ex+1}}+t^{1+\frac1\ex}-t$. By Lemma~\ref{lemma:ex_bound} this is at most $\frac{1}{\ex+1}$ and thus both the maximal and the expected bias go to 0 as $\ex\rightarrow \infty$.

Therefore by selecting sufficiently large $\ex$ and sufficiently small $\lin$, we approach a `perfect' mechanism with no verification and no bias. However we must not forget that the parameters are restricted by the maximal allowed penalty $\xi$. Indeed, ignoring this restriction provides us with the noisy counterpart of the $\hgp$ mechanism. 

\paragraph{From two parameters to one}
By the above discussion, once $\xi$ is fixed and $\ex$ was selected, we should choose the minimal $\lin$ that satisfies \HR3.  Denote 
$\lin^*(\xi,\ex):=\left((1+1/\ex)^{-(\ex+1)}+\ex\cdot\xi\right)^{-1/(\ex+1)}$.
\begin{corollary}
    For any $\xi\geq 0$ and $\kappa \in \mathbb Z_+$, any valid $\PM^{\lin,\ex}$ mechanism is dominated by  $\PM^{\lin^*(\xi,\ex),\ex}$. 
\end{corollary}
This is simply since mechanisms with $\lin<\lin^*(\ex,\xi)$ violate \HR3 (thus not valid), and mechanisms with $\lin>\lin^*(\ex,\xi)$ have higher verification rate and same bias.  

Thus, as in the deterministic verification case, we get a three-way tradeoff: fixing the maximal penalty $\xi$, the (integer) parameter $\ex$ balances  verification vs. bias. However in the PM class, allowing negative grades is essential. We conjecture that under limited liability, no valid mechanism is better than a convex combination of $\verall$ and $\payall$. 

In contrast to the deterministic case, we have no guarantee that the PM class provides the best possible tradeoff. 
However, the  class provides a strong baseline for mechanisms that balance between verification and bias. 

\begin{figure}
    \centering
\newcommand{\gammaval}{2.}
\newcommand{\gammafix}{0.02}
\newcommand{\xival}{0.75}
\newcommand{\linval}{0.8}
\newcommand{\maxt}{6.}
\begin{tikzpicture}[yscale=0.40]
    \draw[->] (0,0) -- (\maxt+.5,0) node[right] {$t$};
    \draw[->] (0,-5.5) -- (0,\maxt+4) node[above] {};

     \draw (\maxt, -0.15) -- (\maxt, 0.15);

        \node[below] at (\maxt, -0.15) {$1$};
                \draw (-0.15, 0) -- (0.15, 0);

                        \draw (-0.15,\maxt) -- (0.15,\maxt);

        \draw (-0.15,4/27*\maxt) -- (0.15,4/27*\maxt);

    \node[left] at (-0.1, 4/27*\maxt) {$\frac{4}{27}$};
        \node[left] at (-0.15, 0) {$0$};
    \node[left] at (-0.15, \maxt) {$1$};
    \node[left] at (-0.15, -0.85*\maxt) {$-\xi$};
    
    \draw[dotted, thick] (0,0) -- (\maxt,\maxt);

\draw[thick, green!60!black, domain=0:\maxt, smooth, variable=\t]  plot ({\t},{(\maxt*4/27)+1/(2*\linval^3)*((\linval*\t*\maxt)^(1/2)+(\linval*\t))});

\draw[thin, double, red, domain=0:\maxt, smooth, variable=\t] 
    plot ({\t}, {\linval*sqrt(\maxt*\t)});
  \fill[red!20, opacity=0.4, domain=0:\maxt, smooth, variable=\t] 
    (0,0) -- plot ({\t}, {\linval*sqrt(\maxt*\t)}) --  (\maxt,0) -- cycle;
\draw[thick, orange, dashed] 
    (0,{(4/27-1/(2*\linval^3))*\maxt}) -- (\maxt, {(4/27-1/(2*\linval^3)+1.5/\linval)*\maxt}); 

\draw[thick, purple, dash dot, domain=0:\maxt, smooth, variable=\t] 
    plot ({\t}, {4/27*\maxt + \t*(\t/\maxt)^(0.5) });
      
      \fill[gray, opacity=0.2, domain=0:\maxt, smooth, variable=\t] 
   plot ({\t},  {4/27*\maxt + \t*(\t/\maxt)^(0.5) }) -- (\maxt,\maxt) -- (0,0) -- cycle;
    
\begin{scope}[shift={(3.2,-2)}] 
    \draw[thick] (-0.3, -3.5) rectangle (2.6, 0.5); 

    \draw[thick, green!60!black] (0,0) -- (0.6,0);
    \node[right] at (0.7,0) {$g_i$ for $\bot$}; 

    \draw[dashed, thick, orange] (0,-0.9) -- (0.6,-0.9);
    \node[right] at (0.7,-0.9) {$g_i$ for $s_i=\hat t_i$};

    \draw[thick, purple, dash dot]
     (0,-1.9) -- (0.6,-1.9);
    \node[right] at (0.7,-1.9) {$\bar g_t$}; 

        \draw[thick, double, red]
     (0,-2.8) -- (0.6,-2.8);
    \node[right] at (0.7,-2.8)  {$q_i$}; 

    \draw[thick, dotted]
     (1.4,-2.8) -- (2,-2.8);
    \node[right] at (2.1,-2.8)  {$t$}; 

\end{scope}
\end{tikzpicture}
\caption{Visualization of the Polynomial Verification mechanism with $\ex=2, \lin=0.5$.  The purple dash-dot line shows the expected grade of a truthful agent. As in Fig.~\ref{fig:MCV}, the shaded areas show the overall verification and bias under a uniform type distribution. }
    \label{fig:PV}
\end{figure}
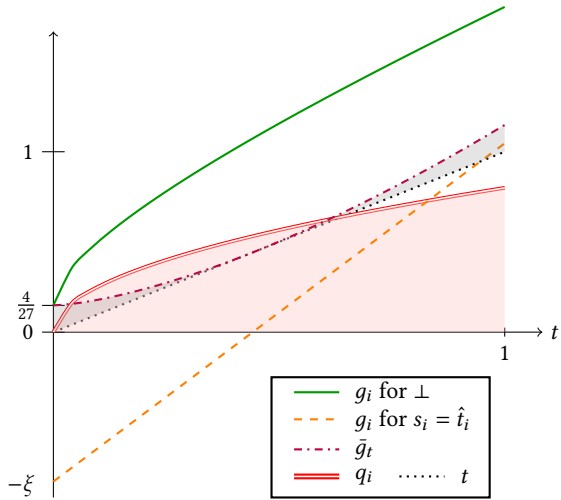
 
\subsection{Alternative Approach:  Type Histogram}\label{sec:det_hist}

The MCV mechanism can work without any information on the type distribution, but can exploit such information to better optimize the bias. 

In this last section, we will assume  the set of possible types is discrete with $|t-t'|>\eps$ for any $t\neq t'$. We further assume we have access to the \emph{exact type histogram} (i.e., we know exactly the number of agents of each type, denoted $H(t):=|\{i\in N: t_i=t\}|$), and show a `toy mechanism' that exploits this information to provide a `perfect' outcome: a unique Nash equilibrium in which \emph{no agent is verified}, and the \emph{maximal bias is arbitrarily small}, even when verification is noisy. 

The equilibrium is not in dominant strategies, and thus the mechanism does not belong in the same class as the previous mechanisms, relaxing \HR1 to \HR1': there is a unique Nash equilibrium, which is truthful. The mechanism also holds \HR2' and limited liability.\footnote{In fact, the mechanism holds a stronger version of \HR2', since truthful agents get their at least type in expectation whether they are verified or not.}

\paragraph{\underline{The Histogram Mechanism:}} The mechanism $\HM^\alpha$ takes a parameter $\alpha\in (1, 1+\eps)$ as input, as well as the histogram $H$ of all agents' types (as ``additional information''). Define $\kappa:=\lceil \frac{2}{\alpha-1}\rceil>\frac2{\eps}$.
\begin{itemize}
    \item Set $\hat H(t):=|\{i:\hat t_i=t\}|$;
    \item $t_V\set \max\{t \in \calT \text{ s.t. }\hat H(t)>H(t)\}$; // (maximal violating type if exists)
      \item Set $\HM^\alpha_V(\hat t_i) = 1$ if $\hat t_i=t_V$ and $0$ otherwise;
    \item $\HM^\alpha_G(\hat t_i,\bot)=\eps +\hat t_i$;
    \item $\HM^\alpha_G(\hat t_i, s)=\eps +R^{1+\frac1\kappa}(\hat t_i,s)$.
\end{itemize}
\begin{proposition}\label{prop:hist}
    The histogram mechanism $\HM^\alpha$ respects limited liability, and has a unique Nash equilibrium, in which all agents are truthful and none are verified. Moreover, for any truthful type~$t$ agent, $\bar g_t\in [t, t+\eps]$ even off-equilibrium. 
\end{proposition}

Intuitively, the proof relies on two observations: first, in every profile where agents only lie upwards, there is at least one cheater that is verified with certainty, and thus experiences the scoring rule fully.
Second, by  Lemma~\ref{lemma:ex_bound} the proper scoring rule $R^{1+\frac1\kappa}$ implements a near-identity function when $\alpha$ is close to 1 (and thus $\kappa$ is large), thus approximating agents' true types to arbitrary precision. 




\section{Empirical Demonstration}\label{sec:sim}

In this section, we demonstrate the triple tradeoff between maximal penalty, average bias and verification for the MCV and PV classes. Since the tradeoff relies on an actual type distribution we employ both synthetic and real-world data.

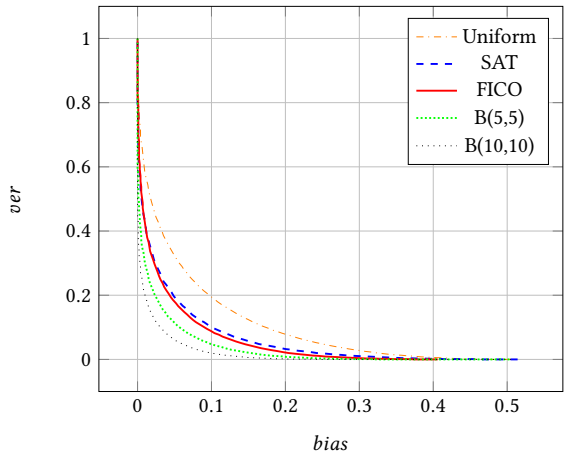
\begin{figure}[t]
    \centering

        \begin{tikzpicture}[scale=0.9]
\begin{axis}[
    xlabel={$bias$},
    ylabel={$ver$},
    grid=both,
    legend pos=north east,
    unbounded coords=discard,
    ymin=-0.1, ymax=1.1,
]


\addplot+[dashdotted, orange, no marks]
table[
    col sep=comma, x=e, y=mu,
]{figures/csv/MCV_xi0_B1_1.csv};
\addlegendentry{Uniform}

\addplot+[thick, dashed, blue, no marks]
table[
    col sep=comma, x=e, y=mu,
]{figures/csv/MCV_xi0_SAT.csv};
\addlegendentry{SAT}

\addplot+[thick, red, no marks]
table[
    col sep=comma,
    x=e, y=mu,
]{figures/csv/MCV_xi0_FICO.csv};
\addlegendentry{FICO}

\addplot+[thick, densely dotted, green, no marks]
table[
    col sep=comma, x=e, y=mu,
]{figures/csv/MCV_xi0_B5_5.csv};
\addlegendentry{B(5,5)}

\addplot+[thin, dotted, black, no marks]
table[
    col sep=comma, x=e, y=mu,
]{figures/csv/MCV_xi0_B10_10.csv};
\addlegendentry{B(10,10)}

\end{axis}
\end{tikzpicture}

    \caption{Efficiency curves for MCV mechanism under \textit{Limited Liability} on several synthetic  and real type distributions. The $\gamma$ parameter ranges from 0 to 1, with higher values to the right.\vspace{-0mm}} 
    \label{fig:MCV_tradeoff}
\end{figure}

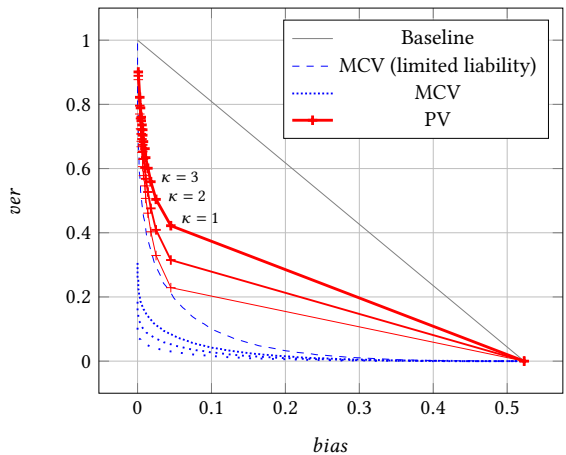
\begin{figure}[t]
    \centering
    \begin{tikzpicture}[scale=0.9]
\begin{axis}[
    xlabel={$bias$},
    ylabel={$ver$},
    grid=both,
    legend pos=north east,
    unbounded coords=discard,
    ymin=-0.1, ymax=1.1,
]


\addplot+[thin, gray, no marks]
table[
    col sep=comma, x=e, y=mu,
]{figures/csv/baseline_SAT.csv};
\addlegendentry{Baseline}

\addplot+[ dashed, blue, no marks]
table[
    col sep=comma, x=e, y=mu,
]{figures/csv/MCV_xi0_SAT.csv};
\addlegendentry{MCV (limited liability)}

\addplot+[thick, densely dotted, blue, no marks]
table[
    col sep=comma, x=e, y=mu,
]{figures/csv/MCV_xi1_SAT.csv};
\addlegendentry{MCV}

\addplot+[thick,  dotted, blue, no marks, forget plot]
table[
    col sep=comma, x=e, y=mu,
]{figures/csv/MCV_xi2_SAT.csv};

\addplot+[thick,  loosely dotted, blue, no marks, forget plot]
table[
    col sep=comma, x=e, y=mu,
]{figures/csv/MCV_xi4_SAT.csv};

\addplot+[very thick, red, mark = +, mark size=2pt]
table[
    col sep=comma, x=e, y=mu,
]{figures/csv/PV_xi1_SAT.csv};
\addlegendentry{PV}

\addplot+[  thick, red, mark = +, mark size=2pt, forget plot]
table[
    col sep=comma, x=e, y=mu,
]{figures/csv/PV_xi2_SAT.csv};

\addplot+[thin, red, mark = +, mark size=2pt, forget plot]
table[
    col sep=comma, x=e, y=mu,
]{figures/csv/PV_xi4_SAT.csv};

\end{axis}
\node at (1.5,2.6) {\scriptsize{$\kappa=1$}};
\node at (1.3,2.92) {\scriptsize{$\kappa=2$}};
\node at (1.2,3.2) {\scriptsize{$\kappa=3$}};

\end{tikzpicture}
        \caption{Efficiency curves of MCV and PV mechanisms on the SAT type distribution. Lines are shown for $\xi=1,2,4$ (and also $\xi=0$ for MCV).  The gray line shows the convex combination of \verall~ and \payall.  \vspace{-0mm}}
    \label{fig:SAT_tradeoff}
\end{figure}

\paragraph{Data} 
As  real-world datasets, we used (a) SAT 2022 test scores: a standardized test widely used for college admissions in the United States, with statistics taken from~\cite{wikipediaSAT}; (b) FICO scores statistics from 2018~\cite{FICO}.

Additionally, we use the synthetic Beta distribution with various parameters. $B(1,1)$ is the Uniform distribution, and higher parameters correspond to lower variance, where MCV is expected to perform better. Both of our real-world datasets are shown in Fig.~\ref{fig:tax}, together with the best-fit Beta distributions. 

\paragraph{Hardware} We used a standard PC with intel Core i7-9700k CPU and 16GB RAM for running the simulations. The entire execution takes a couple of seconds. \rmr{seems redundant...}

\paragraph{Experimental setup}
We set the value of the maximal penalty parameter  $\xi$. 
Note that both mechanism classes essentially have a single parameter: $\gamma\in[0,1]$ for MCV, and $\kappa\in\mathbb Z_+$ for PV. 
For each type distribution $p$, and $M\in\{\M^{\gamma,\xi},\ \PM^{\lin(\xi,\kappa),\kappa}\}$ we compute the $\bias(\M,p)$ and $\ver(\M,p)$ for different values of the parameter in range. 
We select $\gamma$ as multiples of 0.01, and $\kappa\in \{1,\ldots,15,20,50\}$.

We compute $\bias()$ and $\ver()$ as \emph{expected values}, by integrating over the type distribution $p$. We do not sample the types, the verification decision, or (in the noisy setting) the performance of the verified agent.\footnote{In fact, we do not even need to define the performance distribution $T_i$, since we only use its expectation $t_i$.}


\subsection{Results} 

\paragraph{Deterministic verification}
Figure~\ref{fig:MCV_tradeoff}(top) shows
Efficiency curves of the MCV mechanism (i.e., under deterministic verification) for the various distributions, using $\xi=0$ (limited liability). 
Higher $\gamma$ values correspond to points towards the bottom right, i.e., with less verification and a higher bias. 

We can see that MCV performs better as the variance of $p$ decreases, and is doing much better on SAT and FICO than on the Uniform distribution. 

For example, in  SAT, we could guarantee an expected bias of less than 5\% of the grade, with verifying only 20\% of the students. Note that in the noiseless setting, the bias of the MCV mechanism also equals the error in grading (as the grade is deterministic). 

\paragraph{Triple tradeoff}
Fig.~\ref{fig:SAT_tradeoff} compares the performance we get under deterministic verification (using MCV) and noisy verification (using PV), for maximal penalty $\xi\in \{1,2,4\}$. The figure shows results on the SAT distribution, results on other distributions are similar.

We also show results for MCV with limited liability, which still does better than PV with a high maximal penalty. Note that the best we can do in the noisy setting with limited liability is the baseline. 

One thing that is evident from the figure is that we can significantly reduce the bias in both settings, but it is harder to control the verification amount in the noisy setting, as we currently have no solution between the Linear Verification mechanism ($\kappa=1$) and the \payall mechanism.

\section{Conclusion}
In this paper, we considered possible ways to incentivize agents to truly report their type, when a high type leads to better utility (as in the case of grading an exam).   While agents can lie freely about their type, the principal has certain verification power, that she prefers to use scarcely. 

Our main result characterizes the optimal tradeoff between required verification and maximal/expected bias, when verification is guaranteed to reveal the true type.  We note that our mechanism can also be used when agents have multi-dimensional types as in \cite{estornell2021incentivizing}, as long as allocation probability is only determined by the agent's own type (thus can be treated as `grade'), i.e., when there is no competition.

Noisy verification makes the task much more challenging, but still possible if we allow `negative grades', and at the cost of larger bias and/or more verification. We further showed how having access to the exact type histogram leads to mechanism that essentially eliminates both the bias and the required verification. 
While having such access is unlikely, we hope the ideas used in the proof could be useful  to better exploit information on the type distribution. 

The main remaining challenge is to prove a tight approximation bound for the noisy verification setting. 

\begin{acks}
The work of O. Ben{-}Porat was partially supported by the Israel Science Foundation (ISF; Grant No. 3079/24). 
We thank the anonymous reviewers for their valuable feedback.
\end{acks}
\bibliographystyle{ACM-Reference-Format} 
\bibliography{deterrence}

@inproceedings{ceppi2019partial,
  title={Partial verification as a substitute for money},
  author={Ceppi, Sofia and Kash, Ian and Frongillo, Rafael},
  booktitle={Proceedings of the AAAI Conference on Artificial Intelligence},
  volume={33},
  number={01},
  pages={1837--1844},
  year={2019}
}

@inproceedings{witkowski2012robust,
  title={A robust bayesian truth serum for small populations},
  author={Witkowski, Jens and Parkes, David},
  booktitle={Proceedings of the AAAI Conference on Artificial Intelligence},
  volume={26},
  number={1},
  pages={1492--1498},
  year={2012}
}

@article{savage1971elicitation,
  title={Elicitation of personal probabilities and expectations},
  author={Savage, Leonard J},
  journal={Journal of the American Statistical Association},
  volume={66},
  number={336},
  pages={783--801},
  year={1971},
  publisher={Taylor \& Francis}
}

@article{carroll2019robust,
  title={Robust incentives for information acquisition},
  author={Carroll, Gabriel},
  journal={Journal of Economic Theory},
  volume={181},
  pages={382--420},
  year={2019},
  publisher={Elsevier}
}

@inproceedings{epitropou2019optimal,
  title={Optimal on-line allocation rules with verification},
  author={Epitropou, Markos and Vohra, Rakesh},
  booktitle={Algorithmic Game Theory: 12th International Symposium, SAGT 2019, Athens, Greece, September 30--October 3, 2019, Proceedings 12},
  pages={3--17},
  year={2019},
  organization={Springer}
}

@article{ben2014optimal,
  title={Optimal allocation with costly verification},
  author={Ben-Porath, Elchanan and Dekel, Eddie and Lipman, Barton L},
  journal={American Economic Review},
  volume={104},
  number={12},
  pages={3779--3813},
  year={2014},
  publisher={American Economic Association 2014 Broadway, Suite 305, Nashville, TN 37203}
}

@article{mylovanov2017optimal,
  title={Optimal allocation with ex post verification and limited penalties},
  author={Mylovanov, Tymofiy and Zapechelnyuk, Andriy},
  journal={American Economic Review},
  volume={107},
  number={9},
  pages={2666--2694},
  year={2017},
  publisher={American Economic Association 2014 Broadway, Suite 305, Nashville, TN 37203}
}

@article{chua2023optimal,
  title={Optimal multi-unit allocation with costly verification},
  author={Chua, Geoffrey A and Hu, Gaoji and Liu, Fang},
  journal={Social Choice and Welfare},
  volume={61},
  number={3},
  pages={455--488},
  year={2023},
  publisher={Springer}
}

@article{akasiadis2017cooperative,
  title={Cooperative electricity consumption shifting},
  author={Akasiadis, Charilaos and Chalkiadakis, Georgios},
  journal={Sustainable Energy, Grids and Networks},
  volume={9},
  pages={38--58},
  year={2017},
  publisher={Elsevier}
}

@article{selten1998axiomatic,
  title={Axiomatic characterization of the quadratic scoring rule},
  author={Selten, Reinhard},
  journal={Experimental Economics},
  volume={1},
  pages={43--61},
  year={1998},
  publisher={Springer}
}

@article{gneiting2007strictly,
  title={Strictly proper scoring rules, prediction, and estimation},
  author={Gneiting, Tilmann and Raftery, Adrian E},
  journal={Journal of the American statistical Association},
  volume={102},
  number={477},
  pages={359--378},
  year={2007},
  publisher={Taylor \& Francis}
}

@article{prelec2004bayesian,
  title={A Bayesian truth serum for subjective data},
  author={Prelec, Drazen},
  journal={science},
  volume={306},
  number={5695},
  pages={462--466},
  year={2004},
  publisher={American Association for the Advancement of Science}
}

@inproceedings{hardt2016strategic,
  title={Strategic classification},
  author={Hardt, Moritz and Megiddo, Nimrod and Papadimitriou, Christos and Wootters, Mary},
  booktitle={Proceedings of the 2016 ACM conference on innovations in theoretical computer science},
  pages={111--122},
  year={2016}
}

@inproceedings{levanon2021strategic,
  title={Strategic classification made practical},
  author={Levanon, Sagi and Rosenfeld, Nir},
  booktitle={International Conference on Machine Learning},
  pages={6243--6253},
  year={2021},
  organization={PMLR}
}

@article{zrnic2021leads,
  title={Who leads and who follows in strategic classification?},
  author={Zrnic, Tijana and Mazumdar, Eric and Sastry, Shankar and Jordan, Michael},
  journal={Advances in Neural Information Processing Systems},
  volume={34},
  pages={15257--15269},
  year={2021}
}

@article{zohar2008mechanisms,
  title={Mechanisms for information elicitation},
  author={Zohar, Aviv and Rosenschein, Jeffrey S},
  journal={Artificial Intelligence},
  volume={172},
  number={16-17},
  pages={1917--1939},
  year={2008},
  publisher={Elsevier}
}

@article{porter2008fault,
  title={Fault tolerant mechanism design},
  author={Porter, Ryan and Ronen, Amir and Shoham, Yoav and Tennenholtz, Moshe},
  journal={Artificial Intelligence},
  volume={172},
  number={15},
  pages={1783--1799},
  year={2008},
  publisher={Elsevier}
}

@article{meir2012algorithms,
  title={Algorithms for strategyproof classification},
  author={Meir, Reshef and Procaccia, Ariel D and Rosenschein, Jeffrey S},
  journal={Artificial Intelligence},
  volume={186},
  pages={123--156},
  year={2012},
  publisher={Elsevier}
}

@article{nisan2007introduction,
  title={Introduction to mechanism design (for computer scientists)},
  author={Nisan, Noam and others},
  journal={Algorithmic game theory},
  volume={9},
  pages={209--242},
  year={2007}
}

@article{kahneman2003maps,
  title={Maps of bounded rationality: Psychology for behavioral economics},
  author={Kahneman, Daniel},
  journal={American economic review},
  volume={93},
  number={5},
  pages={1449--1475},
  year={2003},
  publisher={American Economic Association}
}

@inproceedings{branzei2015verifiably,
  title={Verifiably truthful mechanisms},
  author={Br{\^a}nzei, Simina and Procaccia, Ariel D},
  booktitle={Proceedings of the 2015 Conference on Innovations in Theoretical Computer Science},
  pages={297--306},
  year={2015}
}

@book{selten1988reexamination,
  title={Reexamination of the perfectness concept for equilibrium points in extensive games},
  author={Selten, Reinhard and Bielefeld, R Selten},
  year={1988},
  publisher={Springer}
}

@article{li2017obviously,
  title={Obviously strategy-proof mechanisms},
  author={Li, Shengwu},
  journal={American Economic Review},
  volume={107},
  number={11},
  pages={3257--3287},
  year={2017},
  publisher={American Economic Association 2014 Broadway, Suite 305, Nashville, TN 37203}
}

@article{ashlagi2018stable,
  title={Stable matching mechanisms are not obviously strategy-proof},
  author={Ashlagi, Itai and Gonczarowski, Yannai A},
  journal={Journal of Economic Theory},
  volume={177},
  pages={405--425},
  year={2018},
  publisher={Elsevier}
}

@misc{ball2025probabilistic_verification,
      title={Probabilistic Verification in Mechanism Design}, 
      author={Ian Ball and Deniz Kattwinkel},
      year={2025},
      eprint={1908.05556},
      archivePrefix={arXiv},
      primaryClass={econ.TH},
      url={https://arxiv.org/abs/1908.05556}, 
}

@article{li2020mechanism,
  title={Mechanism design with costly verification and limited punishments},
  author={Li, Yunan},
  journal={Journal of Economic Theory},
  volume={186},
  pages={105000},
  year={2020},
  publisher={Elsevier}
}

@article{green1986partially,
  title={Partially verifiable information and mechanism design},
  author={Green, Jerry R and Laffont, Jean-Jacques},
  journal={The Review of Economic Studies},
  volume={53},
  number={3},
  pages={447--456},
  year={1986},
  publisher={Wiley-Blackwell}
}

@misc{IRS,
author= {IRS},
year = 2013,
title = {SOI Tax Stats},
note ={\url{ https://www.irs.gov/statistics/soi-tax-stats-individual-statistical-tables-by-tax-rate-and-income-percentile}}
}

@misc{FICO,
year = 2018,
author = {About-Rankings},
title = {Blog Post},
note = {\url{https://aboutranking.com/2018/01/21/how-a-credit-score-is-calculated-and-how-objective-it-is/}}
}

@article{mookherjee1989optimal,
  title={Optimal auditing, insurance, and redistribution},
  author={Mookherjee, Dilip and Png, Ivan},
  journal={The Quarterly Journal of Economics},
  volume={104},
  number={2},
  pages={399--415},
  year={1989},
  publisher={MIT Press}
}

@misc{wikipediaSAT,
  author       = {{Wikipedia contributors}},
  title        = {SAT --- Wikipedia{,} The Free Encyclopedia},
  year         = 2025,
  url          = {https://en.wikipedia.org/wiki/SAT},
  note         = {Accessed: 2025-07-28},
  howpublished = {\url{https://en.wikipedia.org/wiki/SAT}}
}

@inproceedings{estornell2021incentivizing,
  title={Incentivizing truthfulness through audits in strategic classification},
  author={Estornell, Andrew and Das, Sanmay and Vorobeychik, Yevgeniy},
  booktitle={Proceedings of the AAAI Conference on Artificial Intelligence},
  volume={35},
  number={6},
  pages={5347--5354},
  year={2021}
}

@inproceedings{ben2024principal,
  title={Principal-agent reward shaping in mdps},
  author={Ben-Porat, Omer and Mansour, Yishay and Moshkovitz, Michal and Taitler, Boaz},
  booktitle={Proceedings of the AAAI Conference on Artificial Intelligence},
  volume={38},
  number={9},
  pages={9502--9510},
  year={2024}
}

\full{
\clearpage
\onecolumn
\appendix
\section{Omitted Proofs of Section~\ref{sec:det}}

We first observe that
\begin{align*}
   p_t(t) & = \frac{n_t-1}{n-1}=\frac{n}{n-1}p(t)-\frac{1}{n-1}\\
         \forall t' \neq t,\quad p_t(t') & =\frac{n_{t'}}{n-1}=\frac{n}{n-1}p(t')
     \end{align*}
     where $n_t=|\{i \in N:t_i=t\}|$. 
     
\begin{rtheorem}{thm:AS-MCV efficiency}
    For every $\beta \in [0,1]$, and any type dist. $p$,
    \begin{enumerate}
        \item $\bias(\AM^{\beta},p) \leq \beta + \frac{1}{n}$
        \item $\ver(\AM^{\beta},p) \leq \Ver(\gamma(\beta,p),p)$.
    \end{enumerate}
    \end{rtheorem}
    In the proof, we ease the notation of $\gamma_i$ to $\gamma_t$, as in the truthful equilibrium, $\gamma_{i}$ is identical for all agents of (true) type $t_i=t$.   
 \begin{proof}
     We prove for the average error first. Define $\mathcal{D}:=\{t : t \leq \gamma_t\}$. $\mathcal{D}$ contains  all the types that incur distortion in $\AM^{\beta}$, assigned with a score $\gamma_t \geq t$. Thus,
     \[\bias(\AM^{\beta},p)=\sum_{t \in \mathcal{D}}p(t)(\gamma_t-t).\]

   By definition of $\gamma_i$, $\Bias(\gamma_i,p_{-i})\leq \frac{n}{n-1}\beta$ for all $i\in N$.
   Let $\gamma^*:=\max \{\gamma_t : t\in \mathcal{D}\} $ and $t^*$ the corresponding type. Then in particular, $\Bias(\gamma^*,p_{t^*})\leq \frac{n}{n-1}\beta$.
Calculating the error of $\AM^\beta$ over all types, and assuming $t^*\leq \gamma^*$:  
\begin{small}
\begin{align*}
    \bias&(\AM^{\beta},p) =\sum_{t \in \mathcal{D}}p(t)(\gamma_t-t) \leq \sum_{t \leq \gamma^*}p(t)(\gamma^*-t)\\  
     & = \sum_{\mathclap{\substack{t \leq \gamma^* \\ t\neq t^*}}}\frac{n-1}{n}p_{t^*}(t)(\gamma^*-t)
         + \Big(\frac{n-1}{n}p_{t^*}(t^*)+\frac{1}{n}\Big)(\gamma^*-t^*)\\
     & \leq \sum_{t \leq \gamma^*}\frac{n-1}{n}p_{t^*}(t)(\gamma^*-t)+\frac{1}{n}  = \frac{n-1}{n}\Bias(\gamma^*, p_{t^*})+\frac1n=\beta +\frac{1}{n}.
\end{align*}\end{small}
If $t^*>\gamma^*$, we get the same without the $+\frac1n$.
  
  As for expected verified fraction,
  \[\ver(\AM^{\beta},p)=\sum_{t \in \bar{\mathcal{D}}}p(t)\frac{t-\gamma_t}{t+\xi}\]
  where $\bar{\mathcal{D}}=\{t : t > \gamma_t\}$. If $t >\gamma_t$, then all $t'\leq \gamma_t$ differ from $t$ and thus $p_t(t')=\frac{n}{n-1}p(t')$. 
  
  \begin{align*}
      \gamma_{t} & =\max \Big \{\gamma \in [0,1]:\sum_{t' \leq \gamma}p_{t}(t')(\gamma-t') \leq \frac{n}{n-1}\beta \Big \}\\
      & =\max \Big \{\gamma \in [0,1]: \sum_{t' \leq \gamma}\frac{n}{n-1}p(t')(\gamma-t')\leq \frac{n}{n-1}\beta \Big \}\\
       & =\max \Big \{\gamma \in [0,1]: \sum_{t' \leq \gamma}p(t')(\gamma-t')\leq \beta \Big \}
       = \gamma(\beta,p).
  \end{align*}
  i.e. $t > \gamma_t=\gamma(\beta,p)$ for all $t \in \bar{\mathcal{D}}$. Hence, 
  \begin{align*}
      \ver&(\AM^{\beta},p)=\sum_{t \in \bar{\mathcal{D}}}p(t)\frac{t-\gamma_t}{t+\xi}\\
      &\leq \sum_{t > \gamma(\beta,p)}\!p(t)\frac{t-\gamma(\beta,p)}{t+\xi}=\Ver(\gamma(\beta,p),p). \qedhere
  \end{align*}
  
     \end{proof}

     \section{Omitted Proofs of Section~\ref{sec:noisy}}
\subsection{Na\"ive scoring rules fail}\label{apx:no_score}
We show why na\"ively applying a proper scoring rule would not work. Suppose e.g. that the type distribution is uniform, and we apply the proper scoring rule $R^2$, and grade $\hat t+Z$ with some constant $Z$ for unverified agents. 

\begin{proposition} The following holds:
    \begin{itemize}
        \item $q_t$ is weakly increasing in $t$, i.e. higher reported types are verified with higher probability.
        \item $q_0=1$. 
    \end{itemize}
\end{proposition}
Note that this means the only way to apply the scoring rule is by verifying all agents. 
\begin{proof}
Note first that if $q_t$ is discontinuous at some $t'\in(0,1)$, then  $t'+Z \geq t' > R^2(t',t')$, but then there is a manipulation either at $t'+\eps$ to $t'$ or vice versa, where the agent can decrease her inspection probability and increase her grade. Thus $q(t)$ must be continuous.

We write $q(t)$ as a continuous function of $t$. Consider an arbitrary point $t$ where $q(t)<1$ and denote $\delta:=1-q(t)$. 

    The expected utility of an agent truthfully reporting type $t$ is 
    $$g_t=q(t) R^2(t,t) + (1-q(t))(t+Z)= q(t)t^2+(1-q(t))(t+Z)=(t+Z)- q(t) \cdot (t-t^2 + Z).$$
    If the agent reports type $t+\eps$, they get
    \begin{align*}
        g_t(t+\eps)&=q(t+\eps) R^2(t+\eps, t) + (1-q(t+\eps))(t+\eps+Z)=  q(t+\eps) (2t(t+\eps)-(t+\eps)^2) + (1-q(t+\eps))(t+\eps+Z)\\
        &=        q(t+\eps) (t^2 
        -\eps^2) + (1-q(t+\eps))(t+\eps+Z)\\
        &=t+\eps+Z - q(t+\eps)(t-t^2+\eps+\eps^2+Z).
    \end{align*}
    Since the agent is not allowed to gain from reporting $t+\eps$,
    \begin{align*}
        0&\leq g_t-g_t(t+\eps) = ((t+Z)- q(t) \cdot (t-t^2 + Z))-(t+\eps+Z - q(t+\eps)(t-t^2+\eps+\eps^2+Z))\\
        &= q(t+\eps)(t-t^2+\eps+\eps^2+Z) -  q(t) \cdot (t-t^2 + Z) -\eps \\
        &= (q(t+\eps)-q(t))(t-t^2+Z) +  q(t+\eps)(\eps+\eps^2) -\eps &\Rightarrow \\
        q(t+\eps) -q(t)& \geq \frac{\eps-(\eps+\eps^2)q(t+\eps)}{t-t^2+Z} = \frac{\eps(1-(1+\eps)q(t+\eps))}{t-t^2+Z}.
        \end{align*}
        Now, 
        \begin{align*}
            \frac{\partial q(t)}{t} & = \lim_{\eps\rightarrow 0}\frac{q(t+\eps)-q(t)}{\eps} \geq  \lim_{\eps\rightarrow 0} \frac{\eps(1-(1+\eps)q(t+\eps))}{(t-t^2+Z)\eps}\\
            &=  \lim_{\eps\rightarrow 0}\frac{(1-(1+\eps)q(t+\eps))}{t-t^2+Z}.
        \end{align*}

Recall $\delta=1-q(t)>0$. Then by continuity there is $\eps'$ s.t. $q(t')<1-\delta/2$ in the entire interval $[t,t+\eps']$. Let $\eps'':=\min\{\delta/2, \eps'$, then $q(t')<1-\eps''$ in $t'\in [t,t+\eps''] \subseteq [t,t+\eps']$.  

We continue:
   \begin{align*}
       \frac{\partial q(t)}{t}&\geq \lim_{\eps\rightarrow 0}\frac{1-(1+\eps)q(t+\eps)}{t-t^2+Z}>\frac{1-(1+\eps'')\sup_{t'\in[t,t+\eps'']}q(t')}{t-t^2+Z}\\
       &=\frac{1-(1+\eps'')(1-\eps'')}{t-t^2+Z}=\frac{1-(1-(\eps'')^2)}{t-t^2+Z}=\frac{(\eps'')^2}{t-t^2+Z}>0.
   \end{align*}
This means $q(t)$ is strictly increasing everywhere except where $q(t)=1$, and in particular weakly increasing. 

\medskip
Now, observe that  $g_0=Z -q(0)Z$; and
$$g_0(\eps) = \eps + Z - q(\eps)(\eps+\eps^2+Z)\geq \eps+Z+q(0)(\eps+\eps^2+Z)$$

in particular we get
\begin{align*}
0\leq g_0-g_0(\eps)\leq q(0)(\eps+\eps^2)-\eps=\eps(q(0)(1+\eps)-1) \Rightarrow q(0)\geq \frac{1}{1+\eps},
\end{align*}
so in the limit, $q(0)=1$. 
\end{proof}

\subsection{PV mechanism}
     \begin{rlemma}{lemma:PV_g}
    $\bar g_t(\hat t) = \frac{\ex^\ex}{(\ex+1)^{\ex+1}}+  R^{1+\frac{1}{\ex}}(\hat t,  t)$;
\end{rlemma}

\begin{proof}
     \begin{align*}
        \bar g_t(\hat t) &-\frac{\ex^\ex}{(\ex+1)^{\ex+1}} =  \PM^{\lin,\ex}_V(\hat t)E_{s\sim T}[\PM^{\lin,\ex}_G(\hat t, s)-\frac{\ex^\ex}{(\ex+1)^{\ex+1}}] \\
        &~~~+ (1-\PM^{\lin,\ex}_V(\hat t))(\PM^{\lin,\ex}_G(\hat t,\bot)-\frac{\ex^\ex}{(\ex+1)^{\ex+1}})\\
        &=\lin\cdot( \hat t)^{\frac{1}{\ex}}\left(\frac{1+\frac{1}{\ex}}{\lin}E[s]-\frac1{\ex\cdot \lin^{\ex+1}} \right)+ (1-\lin\hat t^{\frac{1}{\ex}})\frac{1}{\ex\lin^{\ex+1}}\sum_{\ell=1}^\ex (\lin\hat t^\frac{1}{\ex})^{\ell}\\  
   &=(\hat t)^{\frac{1}{\ex}}(1+\frac{1}{\ex})t-\frac{(\hat t)^{\frac{1}{\ex}}}{\ex\lin^\ex}+ \frac{1}{\ex\lin^{\ex+1}}\sum_{\ell=1}^\ex (\lin\hat t^\frac{1}{\ex})^{\ell} - \frac{1}{\ex\lin^{\ex+1}}\sum_{\ell=1}^\ex (\lin\hat t^\frac{1}{\ex})^{\ell+1} 
   \\
   &=(\hat t)^{\frac{1}{\ex}}(1+\frac{1}{\ex})t-\frac{(\hat t)^{\frac{1}{\ex}}}{\ex\lin^\ex}+ \frac{1}{\ex\lin^{\ex+1}} (\lin\hat t^\frac{1}{\ex})^{1} - \frac{1}{\ex\lin^{\ex+1}} (\lin\hat t^\frac{1}{\ex})^{\ex+1} \\
   &=(\hat t)^{\frac{1}{\ex}}(1+\frac{1}{\ex})t-\frac{(\hat t)^{\frac{1}{\ex}}}{\ex\lin^\ex}+ \frac{(\hat t)^{\frac{1}{\ex}}}{\ex\lin^\ex} - \frac{1}{\ex} (\hat t^\frac{1}{\ex})^{\ex+1} \\
    &=(1+\frac{1}{\ex})(\hat t)^{\frac{1}{\ex}}t - \frac{1}{\ex} \hat t^{1+\frac{1}{\ex}} 
   \end{align*}
   where the last equality is by setting $\alpha=1+\frac1\ex$ in Eq.~\eqref{eq:R_alpha}.
\end{proof}



\subsection{Histogram Mechanism}
\begin{rproposition}{prop:hist}
    The histogram mechanism $\HM^\alpha$ has a unique Nash equilibrium, in which all agents are truthful and none are verified. Moreover, for any truthful agent~$i$, $E[g_i]\in [t_i, t_i+\eps]$ even off-equilibrium. 
\end{rproposition}

\begin{proof}[Proof of Prop.~\ref{prop:hist}] 
    We first note that the grade range of truthful agents follows directly from Lemma~\ref{lemma:ex_bound} (as $\frac{\kappa^\kappa}{(1+\kappa)^{1+\kappa}}<\frac{1}{\kappa}\leq \eps$).
    As for limited liability, note that at any realization $s\sim T$ 
    $$R^{1+\frac1\kappa}(\hat t,s)=(1+\frac1\kappa)s(\hat t)^{\frac1\kappa}-\frac1\kappa(\hat t)^{1+\frac1\kappa}\geq -\frac1\kappa> -\frac{\eps}{2}>-\eps, $$
    so with the constant addition of $\eps$ the grade is never negative. 

    By the definition of the mechanism, when all agents are truthful then $\hat H=H$, there is no verification and no error in grades. Also, this is an equilibrium, since for any agent $g_i=\eps+t_i$, whereas by reporting any $\hat t_i\neq t_i$ the player will be verified (as $\hat H(\hat t_i)=H(\hat t_i)+1$) and get by truthfulness of proper scoring rules, 
    $$E[g'_i]=\eps+E_{s\sim T_i}[R^{1+\frac1\kappa}(\hat t_i,s)]\leq \eps+E_{s\sim T_i}[R^{1+\frac1\kappa}(t_i,s)]< \eps+ t_i .$$

    We next show that any other profile is not an equilibrium. There are two cases: Case~I is when some player $i$ is reporting $\hat t_i\leq t_i-\eps$. we will show that such an agent will be better off by resorting to truthful reporting, whether they are verified or not, i.e. lying downward is strictly dominated. 

    Indeed, the grade of $i$ when reporting $\hat t_i$ is 
    $$g^*_i:= E[g_i]\leq \eps+\hat t_i < \hat t_i + \eps/2 \leq (t_i-\eps) + \eps/2 < t_i - \eps/2.$$
    
    By resorting to truthful reporting, the agent will  get 
    $$E[g'_i] \in [t_i, \eps+t_i] > g^*_i.$$ 
    Thus $i$ is strictly losing by deviating from the truth, and this cannot be an equilibrium. 

    In the remaining case, all agents report $\hat t_i\geq t_i$, with at least one strict inequality.
    
    In particular, there is an agent that reports $\hat t_i> t_i$.\footnote{Note that we do not care about the magnitude of deviation in this case.} Let $j$ be an agent with maximal such over-report, and let $\ol t$ be this maximal over-reported type. 

   We first observe that $\ol t=t_V$, i.e. the reported type $\ol t$ is the one selected for verification. This holds since no agent with true type $t_i=\ol t$ would lie (cannot lie downward since this would fall under Case~I, and upwards by selection of $\ol t$). Therefore $j$ is added to all agents whose truthful type is $\ol t$, and $\hat H(\ol t)>H(\ol t)$. There is no agent with $t_i>\ol t$ falsely reporting $\hat t_i = \ol t$, since this again would fall under Case~I.

    By being verified, $j$ gets (by properties of a proper scoring rule):
    $$E[g_j|\hat t_j]=\eps+E_{s\sim T_j}[R^{1+\frac1\kappa}(\hat t_j,s)]<\eps+ E_{s\sim T_j}[R^{1+\frac1\kappa}(t_j,s)],$$
    whereas by resorting to truthful report, $j$ would get either $\eps+E_{s\sim T_j}[R^{1+\frac1\kappa}(t_j,s)]$ (if still verified), or $\eps+t_j$ (if not verified) which is in either case  higher than $E[g_j|\hat t_j]$.  This completes the proof.
\end{proof}

}{}

\end{document}